\documentclass[review,12pt]{elsarticle}
\pdfoutput=1
\include{preamble}
\usepackage{color}
\usepackage{enumitem,multicol}
\usepackage{graphicx} 
\usepackage{epsfig} 
\usepackage{amsmath,amsthm} 
\usepackage{amssymb}  
\usepackage{float}
\usepackage{epstopdf}
\usepackage[bitstream-charter]{mathdesign}
\newtheorem{proposition}{Proposition}

\usepackage{hyperref}
\hypersetup{
     pdftoolbar=true,        
     pdfnewwindow=true,      
     colorlinks=true,       
     linkcolor=blue,
     citecolor=blue,        
     filecolor=blue,      
     urlcolor=blue          
}
\begin{document}

\begin{frontmatter}
\title{An Input Reconstruction Approach for Command Following in Linear MIMO Systems\tnoteref{mytitlenote}}
\tnotetext[mytitlenote]{This work was supported by Department of Science and Technology through projects SR/S3/MERC/0064/2012 and SERC/ET-0150/2012, and in part by IIT Gandhinagar.}
 
  \author{R. A.~Chavan\fnref{fn1}}
  \ead{rch258@g.uky.edu}
  \fntext[fn1]{Was at the Indian Institute of Technology, Gandhinagar when the work started, currently a graduate student at the University of Kentucky, Lexington, KY, USA.}
  \address{University of Kentucky, Lexington, KY, USA}
  
  \author{S.~D.~Kadam \fnref{fn2}}
  \ead{kadam.sujay@iitgn.ac.in}
 \fntext[fn2]{Graduate student in Electrical Engineering at the Indian Institute of Technology, Gandhinagar.}
  \address{Indian Institute of Technology Gandhinagar, Palaj, Gandhinagar, Gujarat, India.}

 \author{A.~Rajiv \fnref{fn3}}
  \ead{abhijith.rajiv@gmail.com}
  \fntext[fn3]{Was at the Indian Institute of Technology, Gandhinagar when the work started, currently at Magic Leap, Seattle, WA, USA.}
  \address{Magic Leap, Seattle, WA, USA}
  
  \author{H J.~Palanthandalam-Madapusi \corref{cor1}\fnref{fn4}}
  \ead{harish@iitgn.ac.in}
  \fntext[fn4]{Associate Professor in Mechanical Engineering at the Indian Institute of Technology, Gandhinagar.}
  \address{Indian Institute of Technology Gandhinagar, Palaj, Gandhinagar, Gujarat, India.}
  \cortext[cor1]{Corresponding author}
 
\begin{abstract}
The idea of posing a command following or tracking control problem as an input reconstruction problem is explored in the paper. For a class of square MIMO systems with known dynamics, by pretending that reference commands are actual outputs of the system, input reconstruction methods can be used to determine control action that will result in a system following desired reference commands. A feedback controller which is a combination of an unbiased state estimator and an input reconstructor that ensures unbiased tracking of reference commands is proposed.  Simulations and real-time implementation are presented to demonstrate utility of the proposed idea. Conditions under which proposed controller may be used for non-square systems are also discussed.
\end{abstract}
\end{frontmatter}

\section{Introduction}
The topic of input reconstruction has seen a number of developments recently \cite{xiong03,houautomatica1998,glover69,corless98,kitanidis,steven_kitanidis,palanthUMVACC2007,silverman:69,damato2013}. 
Input reconstruction methods determine the unknown inputs (deterministic) to a system given model information and output measurements originating from those set of unknown inputs. These are also referred to as left inversion problems.
A command following problem can be seen as an input reconstruction problem in a sense that the reference command can be viewed as the outputs of the system and the controller seeks to reconstruct inputs that would yield these desired outputs (reference commands). 

 In that sense, by implicitly assuming that  there exists a control input that yields the desired output $y_{\rm ref}$, input reconstruction can be used  to determine the control inputs that yield the desired outputs by treating the desired outputs are the actual outputs of the system. However, a brute-force left-inversion approach results in a feedforward (open-loop) control, and hence there is a need to integrate a feedback approach with a left inversion approach.
In this paper, we borrow input reconstruction methodologies from previous works and combine them to develop a command following controller based on left inversion that also naturally integrates feedback. An advantage of using such an approach is that, it is readily generalized to MIMO systems as the input reconstruction methods are inherently multivariable. The problems of input reconstruction-left invertibility and tracking control-right invertibility are duals of each other \cite{sain69,marroGA} and in case of left invertible systems it is possible to determine the unknown inputs from system outputs, whereas, in case of right invertible systems, it is possible to generate inputs to track given reference commands. 
Further, it is shown in the paper that the tracking of reference commands by outputs is unbiased for systems with same number of inputs and outputs. Illustrative examples are presented to demonstrate the utility of the suggested control scheme.

The paper is arranged as follows. Section \ref{sec2:ProblemStatement} puts forth the problem of following the desired reference commands in an expectation sense. Methodology to address the defined problem using input reconstruction and state estimation methods is presented in Section \ref{sec3:Methodology}. The control scheme resulting as a combination of an unbiased input reconstructor and a state estimator and remarks on tracking error are presented in Section \ref{cirsem}. 
Illustrative numerical examples to highlight the utility of the proposed scheme are presented in Section \ref{sec6:NumericalResults}. 
Section \ref{sec8:remarks} discusses assumptions on number of plant inputs-outputs and their link with the problem of command following. 
 Section \ref{sec10:conc} provides some concluding remarks.

\section{Problem Statement}\label{sec2:ProblemStatement}
\indent Consider the linear time invariant system with outputs $y_{k+1}$ and with the applied control inputs $\hat{u}_k$ given by
\begin{eqnarray}
x_{k+1} &=& Ax_k + B\hat{u}_k + w_k,\label{p1}\\
y_{k+1} &=& Cx_{k+1} + v_{k+1},\label{p2}
\end{eqnarray}
where $x_k \in \mathbb{R}^n,\, \hat{u}_k \in \mathbb{R}\textcolor{blue}{^p},\, y_k \in
\mathbb{R}^l, w_k \in
\mathbb{R}^n, \text{ and } v_k \in \mathbb{R}^l$. Initially, we assume $l=p$, that is, the system is square. This assumption will be relaxed later to discuss command following using input reconstruction in non-square cases. Let $y_{k+1}$ represent the actual plant output in response to the applied control input $\hat{u}_k$.  The process and sensor noise be denoted by $w_k$ and $v_k$ respectively. These noise sequences are assumed to be i.i.d. Gaussian sequences with zero mean. We assume ${\rm rank}(B)$ = $l$, since in the case of $B$ being rank deficient, one or more inputs are redundant. The system (\ref{p1}), (\ref{p2}) is  assumed to be  state controllable and input and state observable \cite{palanthUMVACC2007}. The assumption of input and state observability further implies ${\rm rank}(CB)=p$. Since $l=p$, ${\rm rank}(CB)$ = $l$ and  the system is trackable \cite{kadam2017revisiting}.

We consider a command following problem in which it is desired that a reference command $y_{\rm ref}$ be followed by the system output.
If $y_{\rm ref}$ is known beforehand, the problem can be seen as a preview control problem. We assume that this reference command can be followed exactly with a (not yet known) desired control input $u_{\rm ref}$ in the noise free case as given by the reference system
\begin{eqnarray}
x_{{\rm ref},k+1} &=& Ax_{{\rm ref},k} + Bu_{{\rm ref},k},\label{ip1}\\
y_{{\rm ref},k+1} &=& Cx_{{\rm ref},k+1},\label{ip2}
\end{eqnarray}
where $x_{\rm ref} \in \mathbb{R}^n$, $u_{\rm ref} \in \mathbb{R}^p$, $ y_{\rm ref} \in
\mathbb{R}^l$, $A \in \mathbb{R}^{n \times n}, B \in \mathbb{R}^{n \times p},$
and $ C \in \mathbb{R}^{l \times n}$. The reference system and the actual plant have the same system matrices $A$, $B$ and $C$ and therefore have the same number of inputs, states and outputs. The reader is reminded that the system considered here is square.

The error in following the reference command is
\begin{align}
y_{{\rm ref},k+1}-y_{k+1}=&~Cx_{{\rm ref},k+1} - Cx_{k+1} - v_k\notag\\
=&~CA(x_{{\rm ref},k} - x_k)+CB(u_{{\rm ref},k}-\hat{u}_k)\notag\\
& -Cw_k-v_k. \label{e2}
\end{align}
Taking the expected value on both sides of (\ref{e2}) yields
\begin{equation} 
\mathbb{E}[y_{{\rm ref},k+1}-y_{k+1}]=CA\mathbb{E}[x_{{\rm ref},k} - x_k] + CB\mathbb{E}[u_{{\rm ref},k}-\hat{u}_k].\label{ce3}
\end{equation}

Equation (\ref{ce3}) implies that the tracking error will be zero in an expectation sense if the terms on the right hand side of (\ref{ce3}) are made zero. The following section discusses the methodology to make the tracking error zero effectively by making the right hand side terms zero.
\section{Methodology} \label{sec3:Methodology}
The tracking error between the reference command and actual output, in an expectation sense is represented by the term on the left hand side of (\ref{ce3}). Looking at the right hand side of (\ref{ce3}), it is logical to approach the command following problem as a two-part exercise, first, to device a strategy to make $\mathbb{E}[u_{{\rm ref},k} - \hat{u}_k]$ zero, and second, to ensure that $\mathbb{E}[x_{{\rm ref},k} - x_k]$ goes to zero. It would be natural to consider an input reconstructor for the first part, and an unbiased state estimator for the second part. We therefore  propose a command following controller that combines a state estimator and an input reconstructor  as shown in Fig. \ref{com_fol_fig}. We analyse the convergence of tracking error and discuss the choices for the state estimator and input reconstructor in the following subsections. 
\subsection{State Estimation} \label{StEst}
Assuming that the applied control input $\hat{u}_{k-1}$ is available and the noise characteristics are known,  unbiased estimates of the actual plant state can be obtained using an optimal estimator. Further, since the system under consideration is linear, the Kalman filter is an obvious choice for the state estimator. With this choice, the estimate $\hat{x}$ of the actual plant state $x$ is
\begin{align}
\hat{x}_{k|k-1} &= A\hat{x}_{k-1|k-1} + B\hat{u}_{k-1},\label{k1}\\
\hat{x}_{k|k}&=  \hat{x}_{k|k-1} + K_{k}(y_{k} - C\hat{x}_{k|k-1`})
\label{k2}
\end{align}
where $y_{k}$ is the known measurement and $\hat{u}_{k-1}$ is the control input already  applied. The Kalman gain $K_{k}$ is computed as
\begin{eqnarray}
P_{{\rm kal},{k|k-1}} &=& AP_{{\rm kal},{k-1|k-1}}A^\intercal + Q,\label{k3}\\
P_{{\rm kal},{k|k}} &=&  (I - K_{k}C) P_{{\rm kal},{k|k-1}},\label{k6}\\
S_{k} &=& C  P_{{\rm kal},{k|k-1}} C ^\intercal + R,\label{k4}\\
K_{k} &=&  P_{{\rm kal},{k|k-1}} C ^\intercal S_{k}^{-1},\label{k5}
\end{eqnarray}
where $P_{\rm kal}$ is the state error covariance of the Kalman filter. 
It must be noted that this choice of the state estimator only gives an unbiased estimate of the actual plant state, that is $\mathbb{E}[\hat{x}_k]=\mathbb{E}[x_k]$ and does not immediately imply that $\mathbb{E}[x_{{\rm ref},k}-x_k] = 0$.  Conditions under which  $\mathbb{E}[\hat{x}_k]=\mathbb{E}[x_k]$ leads to $\mathbb{E}[x_{{\rm ref},k}-x_k] = 0$, and subsequently $\mathbb{E}[y_{{\rm ref},k}-y_k] = 0$, will be brought up later in Section \ref{cirsem}.  Once $\hat{x}_{k|k}$ is estimated using the Kalman filter and known past inputs $\hat{u}_{k-1}$ and current measurement $y_k$, the next step is to determine the control input $\hat{u}_k$ to be applied in current time step. This is discussed in the following subsection.  

\subsection{Input Reconstruction} \label{cirse}
Having chosen Kalman filter as the state estimator in Section \ref{StEst}, the next objective is to choose a suitable input reconstructor that will
provide an unbiased estimate of the desired input $u_{{\rm   ref},k}$. Out of the input reconstructors developed in the literature, we adopt a
filter based input reconstruction method developed in \cite{palanthUMVACC2007} due to its simplicity and inherent ability to handle MIMO systems. This Unbiased Minimum Variance (UMV) filter is closely related to a Kalman filter but has an additional input reconstruction equation and a modified gain to account for the unknown inputs. Input reconstruction for (\ref{p1}), (\ref{p2}) using the UMV is achieved in a three step process \cite{palanthUMVACC2007} given by
\begin{align}
\hat{x}_{k+1|k} &= A \hat{x}_{k|k}, \label{ir1} \\
\hat{x}_{k+1|k+1}
&=  \hat{x}_{k+1|k} + L_{k+1}(y_{k+1} - C\hat{x}_{k+1|k}),
\label{ir2} \\
\hat{u}_{k}  &= B^\dag L_{k+1}(y_{k+1} - C\hat{x}_{k+1|k}), 
\label{ir3}
\end{align}
where $^\dag$ denotes the Moore-Penrose generalized inverse. Here $L_{k+1}$ is the UMV gain obtained by a constrained minimization
of the state error covariance and is given by
\begin{align}
L_{k+1} = B\Pi_{k} + F_{k+1}\tilde{R}_{k+1}^{-1}(I-V\Pi_{k}),
\label{ir4}
\end{align}
where 
\begin{align}
\Pi_{k} & \stackrel{\triangle}{=}  (V^\intercal \tilde{R}_{k+1}^{-1}V)^{-1}V^\intercal \tilde{R}_{k+1}^{-1},
\label{ir5} \\
\hspace{-0.3in}\tilde{R}_{k+1}& \stackrel{\triangle}{=}  CP_{k+1|k}C^\intercal + R, 
\label{ir6}\\
\hspace{-0.3in}P_{k+1|k} & \stackrel{\triangle}{=}  AP_{k|k}A^\intercal + Q, 
\label{ir7} \\
\hspace{-0.3in}P_{k+1|k+1} & \stackrel{\triangle}{=}  P_{k+1|k} - F_{k+1} \tilde{R}_{k+1}^{-1}F_{k+1}^\intercal, 
\label{ir8} \\
\hspace{-0.3in}F_{k+1} & \stackrel{\triangle}{=}  P_{k+1|k}C^\intercal ,
\label{ir9} \\
\hspace{-0.3in} V & \stackrel{\triangle}{=} CB. \label{ir10}
\end{align}
\section{Command following using Input Reconstruction (CIR)}
\label{cirsem}
We next discuss the feedback control scheme that combines the state estimator and the input reconstructor discussed earlier,  for addressing  the command following problem. The proposed scheme  shown in Fig. \ref{com_fol_fig} is referred to as Command following using Input Reconstruction (CIR). 

For generating control inputs, the proposed controller makes use of equations (\ref{ir1}),  (\ref{ir2}) and \begin{equation}\label{cir3}
\hat{u}_{k}  = B^\dag L_{k+1}(y_{{\rm ref},{k+1}} - y_{{\rm pred},k+1}), 
\end{equation}
where $y_{{\rm pred},k+1}=C \hat{x}_{k+1|k}$ is a one-step ahead prediction of the system's output computed by using a one-step open-loop prediction. 
The Kalman filter described by (\ref{k1}) - (\ref{k5}) provides an estimate $\hat{x}_{k|k}$ of the system state $x_k$, using the measured output $y_k$ and control input at previous time instant $\hat{u}_{k-1}$.  The state estimate $\hat{x}_{k|k}$ is used to generate a one-step ahead prediction  $y_{{\rm pred,}{k+1}}=C \hat{x}_{k+1|k}$ of the plant output $y_{k+1}$  which is not available. The already known reference command $y_{{\rm ref},k+1}$,  together with one-step ahead prediction of the plant output  $y_{{\rm pred,}{k+1}}$ and $L_{k+1}$ computed from (\ref{ir4}) are used to determine control input $\hat{u}_k$ at the current time instant using (\ref{cir3}).
Given that measurement and process noises are Gaussian i.i.d. sequences, the Kalman filter provides unbiased, minimum variance estimate of the plant state $\hat{x}_{k+1|k+1}$. The accuracy of  state estimate determines the accuracy of the predicted output $y_{{\rm pred},k+1}$ and in turn the closeness of control input estimates.

\subsection{Analysis of the tracking error}

Since the UMV filter provides unbiased estimates $\hat{u}$ of the desired input ${u_{{\rm ref}}}$, (\ref{ce3}) can be reduced to
\begin{eqnarray}
\mathbb{E}[y_{{\rm ref},k+1}-y_{k+1}]=CA\mathbb{E}[x_{{\rm ref},k} - x_k].\label{ce4}
\end{eqnarray}
We recall from \cite{palanthUMVACC2007} that $L_{k+1}$ given in {\ref{ir2}} satisfies $L_{k+1}CB=B$. A result to show that tracking error in (\ref{ce4}) converges to zero in an expectation sense is now presented.

\begin{proposition}
\label{prop:1}
 Let $\hat{u}_k$ from (\ref{cir3}) and $\hat{x}_{k|k}$ from (\ref{k2}) be such
 that $\mathbb{E}[\hat{u}_k]=\mathbb{E}[u_{{\rm ref},k}]$ and
 $\mathbb{E}[\hat{x}_{k|k}]=\mathbb{E}[x_k]$, respectively. Then,
 \begin{equation}
 \label{propeq}
 \mathbb{E}[y_{{\rm ref},k+1}-y_{k+1}]= 0. 
 \end{equation}
\end{proposition}
\begin{proof}
 Substituting (\ref{ir1}), (\ref{ir2}), (\ref{ip1}) and (\ref{ip2}) in  equation (\ref{ir3}) yields
\begin{align}
\hat{u}_{k}  &= B^\dag L_{k+1}(CAx_{{\rm ref},k} + CBu_{{\rm ref},k} -
CA\hat{x}_{k|k})\notag \\
&= B^\dag L_{k+1}CBu_{{\rm ref},k} + B^\dag L_{k+1}CA(x_{{\rm ref},k} - \hat{x}_{k|k}).\label{sq2} 
\end{align}
Noting that $L_{k+1}CB=B$ and $B^\dag B = I_p$, (\ref{sq2}) simplifies to
\begin{align}
(\hat{u}_{k} - u_{{\rm ref},k})&= B^\dag L_{k+1}CA(x_{{\rm ref},k} -
\hat{x}_{k|k}).\label{sq3}
\end{align}
Taking expected value on both sides of (\ref{sq3}) and noting that
$\mathbb{E}[\hat{u}_{k} - u_{{\rm ref},k}]=0$, we have
\begin{align}
B^\dag L_{k+1}CA\mathbb{E}[x_{{\rm ref},k} - \hat{x}_{k|k}]&=0.\label{sq4}
\end{align}
Next, since $B ^\dagger L_{k+1}CB = B ^\dagger B=I_p$, it follows that ${\rm rank}(B ^\dagger L_{k+1})=l=p$ and from (\ref{sq4}), we have
\begin{align}
CA\mathbb{E}[x_{{\rm ref},k} - \hat{x}_{k|k}]&=0.\label{sq5} 
\end{align}
Further, since $\mathbb{E}[\hat{x}_{k|k}]=\mathbb{E}[x_k]$, it follows that
\begin{align}
CA\mathbb{E}[x_{{\rm ref},k} - x_k]&=0.\label{sq6} 
\end{align}
Substituting (\ref{sq6}) in (\ref{ce3}) we arrive at (\ref{propeq}).
\end{proof}
The results stated above in Proposition 1 holds true for any choice of input reconstructor and state estimator which ensures $\mathbb{E}[\hat{u}_k]=\mathbb{E}[u_{{\rm ref},k}]$ and  $\mathbb{E}[\hat{x}_{k|k}]=\mathbb{E}[x_k]$, respectively.

CIR is a system inversion based control scheme. It is well known that inversion based control schemes do not guarantee bounded and causal control inputs for systems with non-minimum phase zeros. Further, incorporation of a stabilizing state or output feedback does not alleviate the effects of non-minimum phase zero. The use of CIR scheme therefore does not guarantee the existence of bounded control inputs for command following in case of systems with non-minimum phase zeros. 
\section{Numerical Results} \label{sec6:NumericalResults}

\subsection*{\bf Example 1}
\label{ex1}
Consider a two mass spring damper system given by
\begin{align}
\begin{bmatrix}
\dot{x}_{1} \\
\ddot{x}_{1} \\
\dot{x}_{2} \\
\ddot{x}_{2} \\
\end{bmatrix}
= &
\begin{bmatrix}
0 & 1 & 0 & 0\\
\frac{-(k_1 + k_2)}{m_1} & \frac{-(b_1 + b_2)}{m_1} & \frac{k_2}{m_1} & \frac{b_2}{m_1} \\
0 & 0 & 0 & 1\\
\frac{k_2}{m_2} & \frac{b_2}{m_2} & \frac{-k_2}{m_2} & \frac{-b_2}{m_2} \\
\end{bmatrix}
\begin{bmatrix}
x_1 \\
\dot{x}_{1} \\
x_2 \\
\dot{x}_{2} \\
\end{bmatrix} 
 +
\begin{bmatrix}
0 & 0 \\
\frac{1}{m_2} & 0 \\
0 & 0 \\
0 & \frac{1}{m_2} \\
\end{bmatrix}
\begin{bmatrix}
 u_1\\
 u_2\\
\end{bmatrix},\\
\label{eg1} \\
\begin{bmatrix}
y_1 \\
y_2 \\
\end{bmatrix}
= &
\begin{bmatrix}
0 & 1 & 0 & 0\\
0 & 0 & 0 & 1\\
\end{bmatrix}
\begin{bmatrix}
x_1 \\
\dot{x}_{1} \\
x_2 \\
\dot{x}_{2} \\
\end{bmatrix},
\label{eg2}
\end{align}\\
where $m_1=m_2=1$, $k_1 = 4$, $k_2 = 8$, $b_1 = 2,$ and $b_2 = 4$. 
The simulation result of command following performance is shown in Fig. \ref{run} for the proposed algorithm when a sawtooth  and a sinusoidal reference commands are issued to $y_1$ and $y_2$ respectively. The same simulation was run for 100 times and  the mean of the tracking results was observed as shown in Fig. \ref{avg_run}. 
Comparing Fig. \ref{run} with Fig. \ref{avg_run}, it is clear that the command following is unbiased and Proposition \ref{prop:1} is verified.

Tracking performance obtained with proposed CIR scheme is further compared with LQG and MPC controllers tuned at nominal values. The plot of comparison is shown in Fig. \ref{cmstepci}. Table \ref{table:comp1} shows the comparison of mean squared errors for each control scheme.

\subsection*{\bf Example 2}

Next, we consider a MIMO second order RC circuit with two input voltages and two output voltages  as shown in Fig. \ref{fig:linear_MIMO_RC}. The  state space description for this system in continuous time can be written as
\begin{equation}\label{eq:linear_mimo_states1}
\begin{split}
 \begin{bmatrix}
    \frac{{d V _{C_1} (t)}}{dt}\\
    \frac{{d V _{C_2} (t)}}{dt}
  \end{bmatrix}
  &~=~\begin{bmatrix}
    \frac{-(R_1 + R_3)}{C_1 R_1 R_3} & \frac{1}{C_1R_3} \\
    \frac{1}{C_2 R_3} &\frac{-(R_2+R_3)}{C_2R_2R_3}
  \end{bmatrix}
                        \begin{bmatrix}
                          V_{C_1} (t)\\
                          V_{C_2} (t)
                        \end{bmatrix}+ \\
  &\begin{bmatrix}
    \frac{1}{C_1R_1} & 0\\
    0 & \frac{1}{C_2R_2}
  \end{bmatrix}
  \begin{bmatrix}V_{{\rm in}_1}(t)\\V_{{\rm in}_2}(t)\end{bmatrix},
  \end{split}
\end{equation}
\begin{align}
  \begin{bmatrix}V_{{\rm out}_1}(t)\\V_{{\rm out}_2}(t)\end{bmatrix}&~=~\begin{bmatrix}
    1 & 0\\0 & 1
  \end{bmatrix}\begin{bmatrix}
V_{C_1}(t)\\V_{C_2}(t)
\end{bmatrix},\label{eq:linear_mimo_outputs2}
\end{align}
where $V_{C_1}$ and $V_{C_2}$ are the voltages across capacitors $C_1$ and $C_2$ respectively and are also the states of the system. $V_{{\rm in}_1}(t)$, $V_{{\rm in}_2}(t)$ are the input voltages and $R_1$, $R_2$ are the resistances. The objective here is to track reference
commands specified for the output voltages $V_{{\rm out}_1}(t) =V_{C_1} (t)$ and $V_{{\rm out}_2}(t) =V_{C_2} (t)$. To achieve this objective the CIR scheme is implemented in real time as shown in the Fig. \ref{fig:physical_setup}. An Arduino-Uno board is used for communication between the RC circuit and MATLAB-Simulink environment where the code for CIR is executed.  For digital implementation, the continuous time state space model is discretized at a time step of $0.1 s$ for use with CIR scheme. 

The values of the resistances and capacitors used are, $R_1=1 \times 10^{3} ~\Omega$, $R_2=1 \times 10^{3}~ \Omega$, $R_3=1\times10^{3}$, $C_1=1 \times
10^{-6}~\rm F$ and $C_2=330 \times 10^{-6}~ \rm F$. 
Fig. \ref{fig:rc_MIMO} shows the real-time tracking performance of
CIR for different reference commands issued for $V_{C_1}$ and $V_{C_2}$.

\section{Remarks on assumptions} \label{sec8:remarks}
The assumption that the system is square enabled us to show in Proposition \ref{prop:1} that the expected tracking error converges to zero if the estimates of the states and inputs are unbiased. The suggested CIR scheme can be used for following commands in an expectation sense with the Kalman and UMV filters being valid choices for unbiased state estimator and unbiased input reconstructor. In case of non-square systems however, it is not guaranteed that the expected tracking error will converge to zero when CIR scheme is used. In what follows, we discuss how CIR scheme can be used in case of non-square systems to follow reference commands in an expectation sense under some circumstances. 
We discuss the use of CIR scheme for systems with $l<p$ first. 

Given a system 
\begin{eqnarray}
x_{k+1} &=& Ax_k + B{u}_k, \label{mp1}\\
y_{k} &=& Cx_{k}, \label{mp2}
\end{eqnarray} 
with $l<p$, let $N \in \mathbb{R}^{n \times l}$
 be such that $N$ modifies (\ref{mp1}), (\ref{mp2}) as \begin{eqnarray}
x_{k+1} &=& Ax_k + \tilde{B}\tilde{u}_k, \label{mp1_}\\
y_{k} &=& Cx_{k}. \label{mp2_}
\end{eqnarray}
%
where $\tilde{B}=BN \in \mathbb{R}^{n \times l}$ and $\tilde{u}_k =N ^{\dagger} u_k$. 
Suppose $N$ is chosen such that ${\rm rank}(N)=l$ and columns of $N$ belong to row space of $CB$. Then, if (\ref{mp1}) and (\ref{mp2}) is trackable, then (\ref{mp1_}) and (\ref{mp2_}) is trackable as well. Thus, CIR scheme when used on the modified system generates inputs that can be used for following reference commands on the original system when multiplied by $N$ matrix.
Any right inverse of $CB$ such that ${\rm rank} ((CB)^R)=l$, qualifies to be a valid $N$ matrix. 
One such convenient choice is $(CB)^{\dagger}$.

Next, to discuss the use of CIR in the case of systems with $l>p$, we recall a few observations from \cite{kadam2017revisiting}. 
A batch equation for system described by equations (\ref{mp1}) and (\ref{mp2}) for $r \in \mathbb{Z}^+$ samples can be written as

\begin{equation}
\mathcal{Y}_{{r}}= \Gamma_r x_0 + M_r {\mathcal{U}}_{r-1}
\end{equation}

where $\mathcal{Y}_r \triangleq \begin{bmatrix}
y_1 \\ y_2 \\ \vdots \\ y_r \end{bmatrix}$ and ${\mathcal{U}}_{r-1} \triangleq \begin{bmatrix}
{u}_0 \\ {u}_2 \\ \vdots \\ {u}_r \end{bmatrix}$. Also, the matrices $\Gamma_r \in \mathbb{R}^{rl \times n}$ and $M_r \in \mathbb{R}^{rl \times rp}$ are defined as
\begin{equation}\label{GrMr}
\Gamma_r \triangleq \begin{bmatrix}
CA\\
CA^2\\
\vdots \\
CA^{r}
\end{bmatrix}~\text{and}~M_r \triangleq \begin{bmatrix}
CB &0  &\cdots &0 \\
CAB &CB &\cdots &0 \\
\vdots &\vdots  &\ddots &0 \\
CA^{r-1}B &CA^{r-2}B &\cdots &CB\\
\end{bmatrix}.
\end{equation}
In the $l>p$ case it has been established that there exist $\mathcal{Y}_{{\rm ref},r} \notin \mathcal{R}(M_r)$ that cannot be  tracked exactly. In this case however, it is possible to track the sequence $\mathcal{Y}_{{\rm ref},{r}} \Pi_{\mathcal{R}({M_r})} = M_r(M_r^{\intercal} M_r)^{\dagger}M_r^{\intercal} \mathcal{Y}_{{\rm ref,}r}$ which is the orthogonal projection of $\mathcal{Y}_{{\rm ref},{r}} $ on $\mathcal{R}(M_r)$ and therefore is the sequence closest  to $\mathcal{Y}_{{\rm ref},r}$ among all the sequences present in $\mathcal{R}(M_r)$. Thus, CIR scheme can be used to track the modified reference command. This can be seen from the following example.

Consider a system with $A=\begin{bmatrix}
    0.1 &  -0.7 &         0        & 0 \\
    0.7    & 0.2  & -0.7    &0 \\
         0    & 0.7    & 0.3   & -0.7 \\
         0         &0   & 0.7   & 0.4
\end{bmatrix}, B=\begin{bmatrix}
0 \\ 1 \\0 \\0
\end{bmatrix}$ and $ C=\begin{bmatrix}
     0     &1     &0     &0\\
     1     &0     &0     &0
\end{bmatrix}.$ We consider tracking performance under the influence of process noise ($w_k$) and measurement noise ($v_k$) with variances 0.01. The eigenvalues of matrix $A$ are $\lbrace 0.25 \pm 1.211 i,~ 0.25 \pm 0.4338 i\rbrace$. Also, the system has minimum phase zeros at $\lbrace    0.1, 0.35 \pm 0.6982 i \rbrace$. The system can track the projections of the reference command on $\mathcal{R}(M_r)$ in an expectation sense, when CIR is used  and the projected sequence is given as a reference, see Fig. \ref{oa1}.

Alternatively, it is also worthwhile  to note that, if $p-l$ measurements are ignored to make the system a square system such that ${\rm rank}(CB)=l$, then the remaining outputs can be tracked in an expectations sense. 
Removing the second row of $C$ to ignore the second component of the output vector, we have $C=\begin{bmatrix}
0 & 1 & 0 &0
\end{bmatrix}$ and therefore $CB=1$. The tracking performance can be seen in Fig. \ref{oa2}.
\newpage
\section{Conclusion} \label{sec10:conc}

A feedback control scheme for command following in input and state observable square MIMO systems was discussed in this paper. This proposed scheme is based on input reconstruction methods and is akin to left inversion with feedback. The command following problem is reduced to a two part exercise of state estimation and input reconstruction. It was shown that tracking of reference commands is unbiased, if both the state estimator an the input reconstructor are chosen to be unbiased. Simulations showing unbiasedness property were presented along with a real-time implementation using a low-cost hardware. Kalman filter and Unbiased Minimum Variance filter were used for state estimation and  input reconstruction respectively.  Use of the proposed scheme under certain conditions for non-square systems  was also discussed. For systems with more inputs than outputs ($l<p$), the use of $N$ matrix that modified the system to a square one and allowed the use of proposed controller was suggested. In case of systems with more number of outputs than inputs ($l>p$), use of projections and disregarding output measurements to make system square was suggested. 

%
\newpage
\begin{table}
 \caption{Comparison of mean squared errors (MSE) for Example 1} 
 \centering 
 \begin{tabular}{c c c} 
 \hline\hline 
 Filter & \multicolumn{2}{c}{MSE}\\ 
 {} & Output 1 & Output 2\\
 \hline 
 UMV & 0.3397 & 0.3392 \\
 \hline 
LQG & 21.911 & 2.4104 \\
 \hline 
 DMPC & 12.3520 & 4.9201\\
 \hline
 \end{tabular}
 \label{table:comp1} 
 \end{table}

\begin{figure}
\centering
\includegraphics[width=0.7\linewidth]{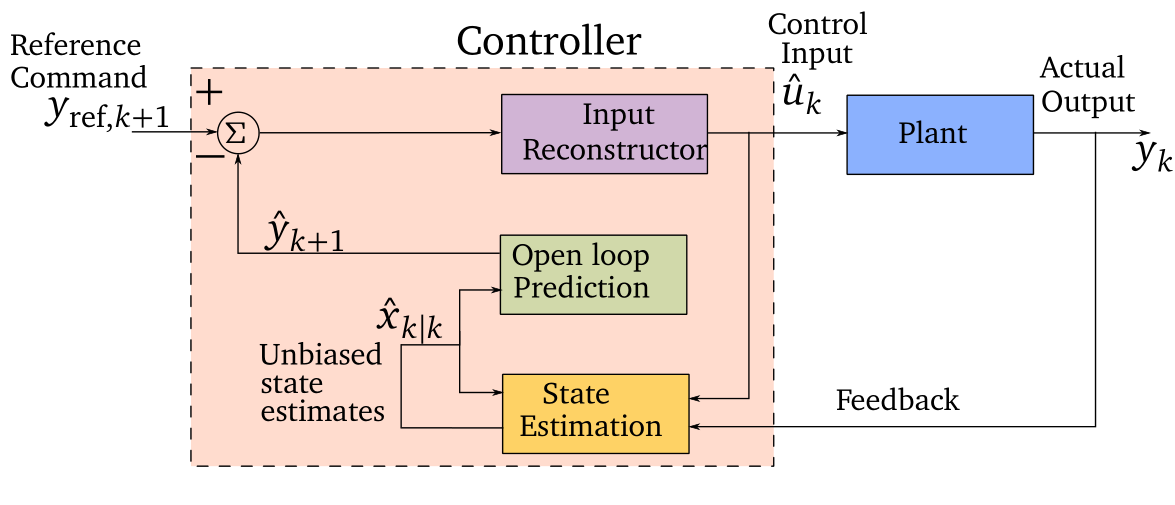}
\caption{Schematic of command following controller.}
\label{com_fol_fig}
\end{figure}
\begin{figure}
\centering

\includegraphics[width=0.55\linewidth]{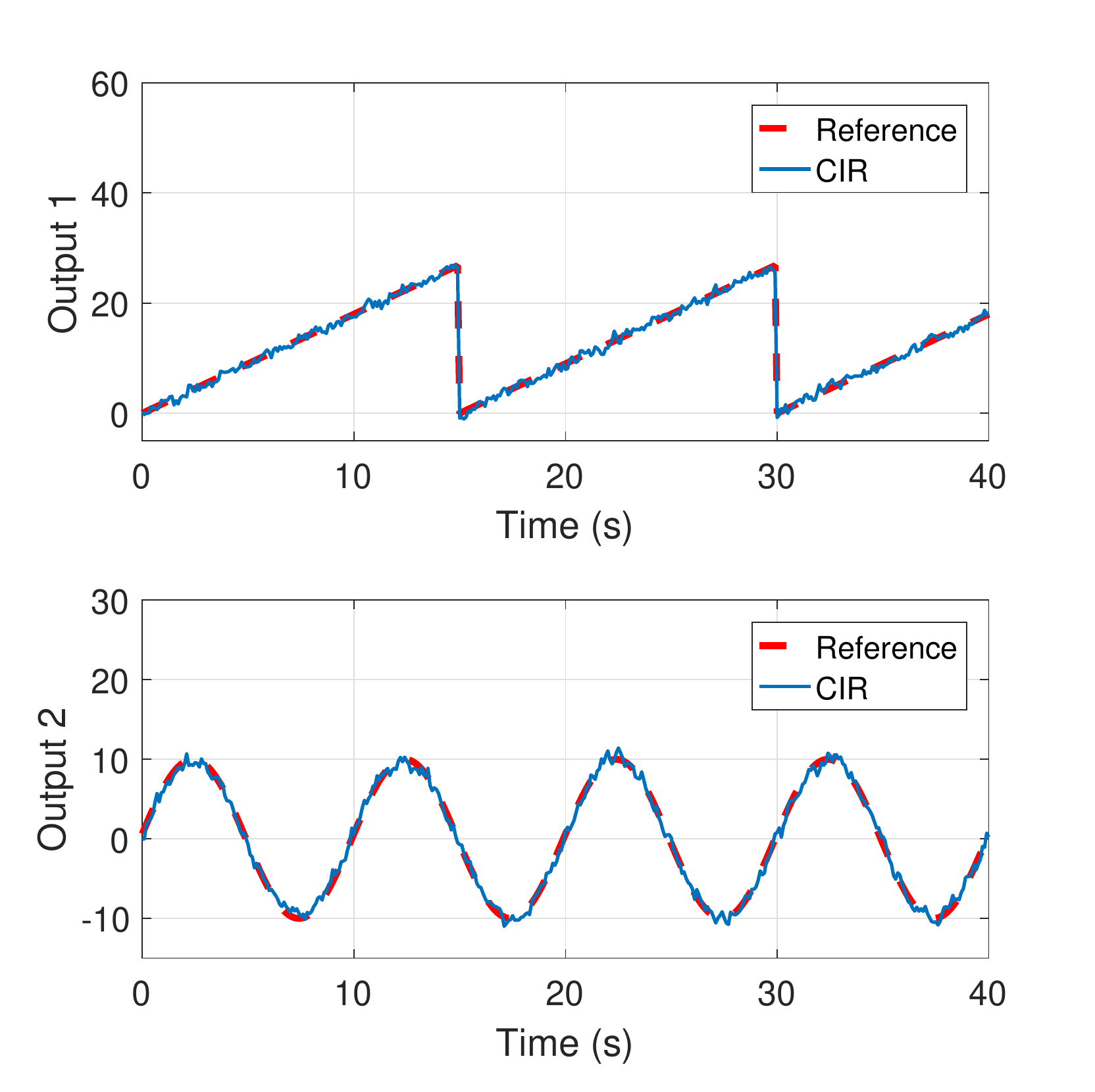}
\caption{Command following for a set of references}
\label{run}
\end{figure}

\begin{figure}
\centering
\includegraphics[width=0.55\linewidth]{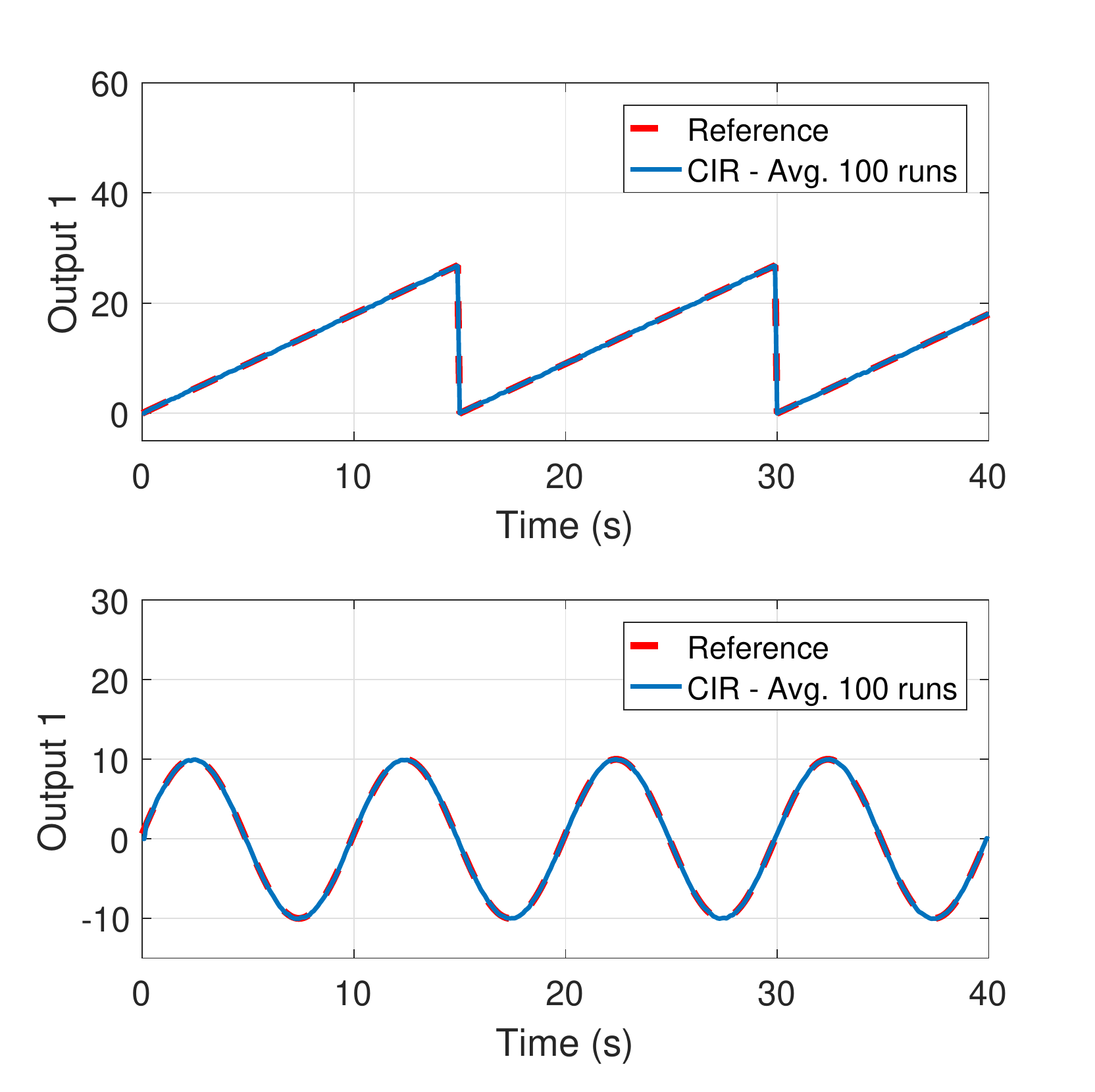}
\caption{Command following for a set of reference commands averaged over 100
runs}
\label{avg_run}
\end{figure}

\begin{figure}
 \centering
 \includegraphics[width=0.55\linewidth]{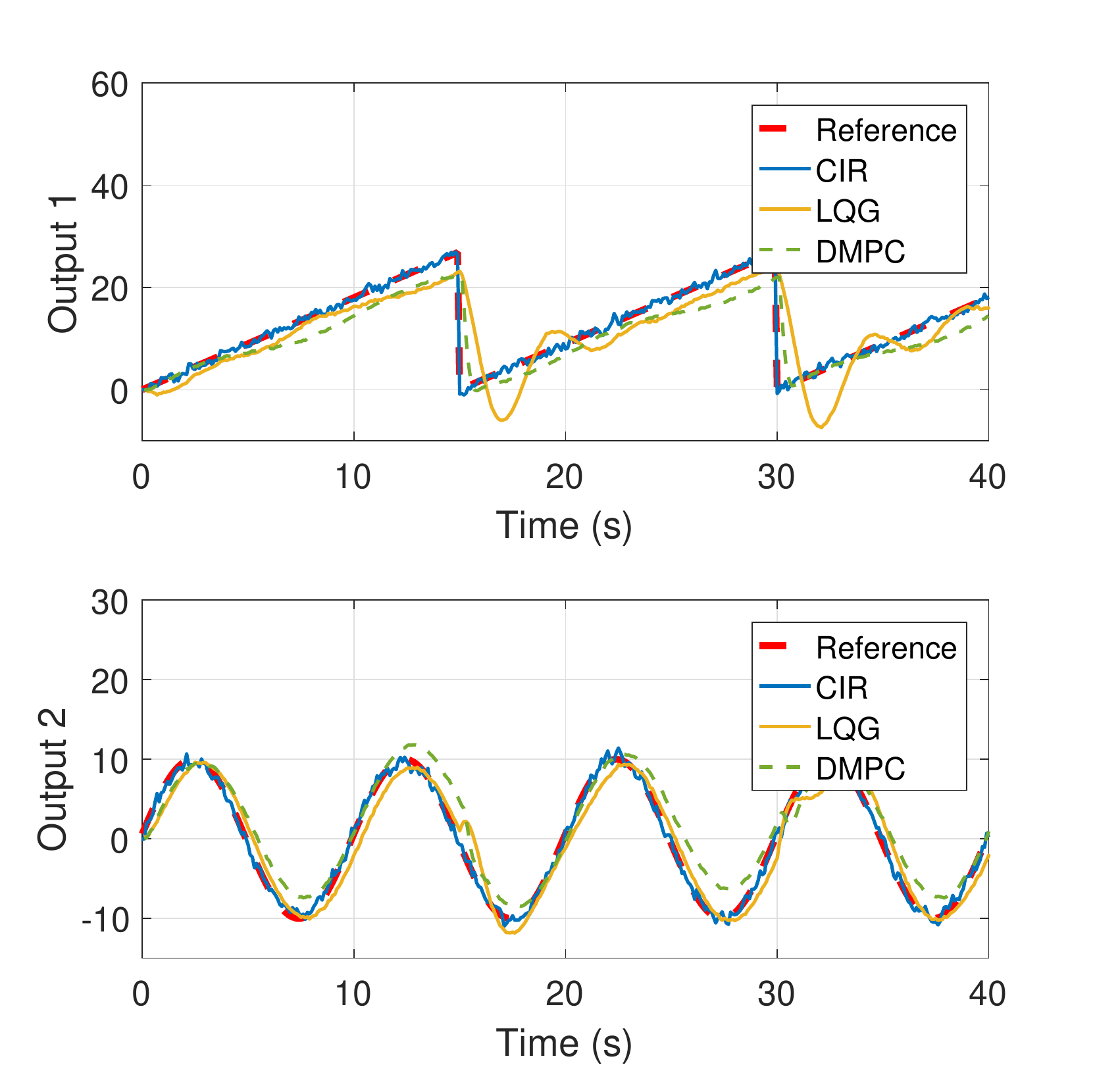}
 \caption{Comparison with LQG and DMPC}
 \label{cmstepci}
 \end{figure}
 
\begin{figure}[!ht]
\centering
\includegraphics[width=0.7\linewidth]{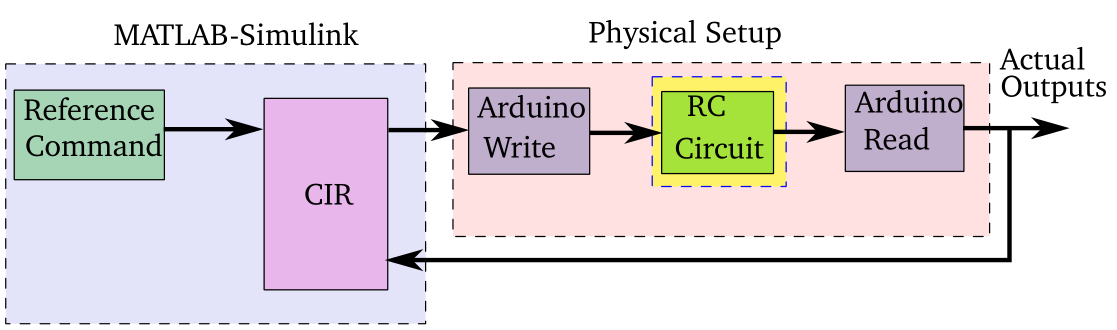}
\caption{Setup for real time command following using CIR}
\label{fig:physical_setup}
\end{figure}

%
%

\begin{figure}
\centering
\includegraphics[width=0.7\linewidth]{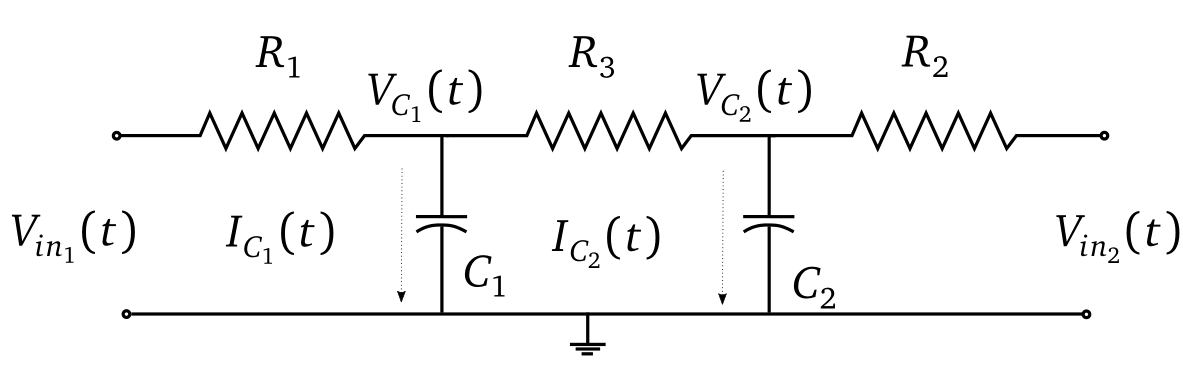}
\caption{A second order MIMO RC circuit}
\label{fig:linear_MIMO_RC}
\end{figure}

\begin{figure}
\centering
\includegraphics[width=0.8\linewidth]{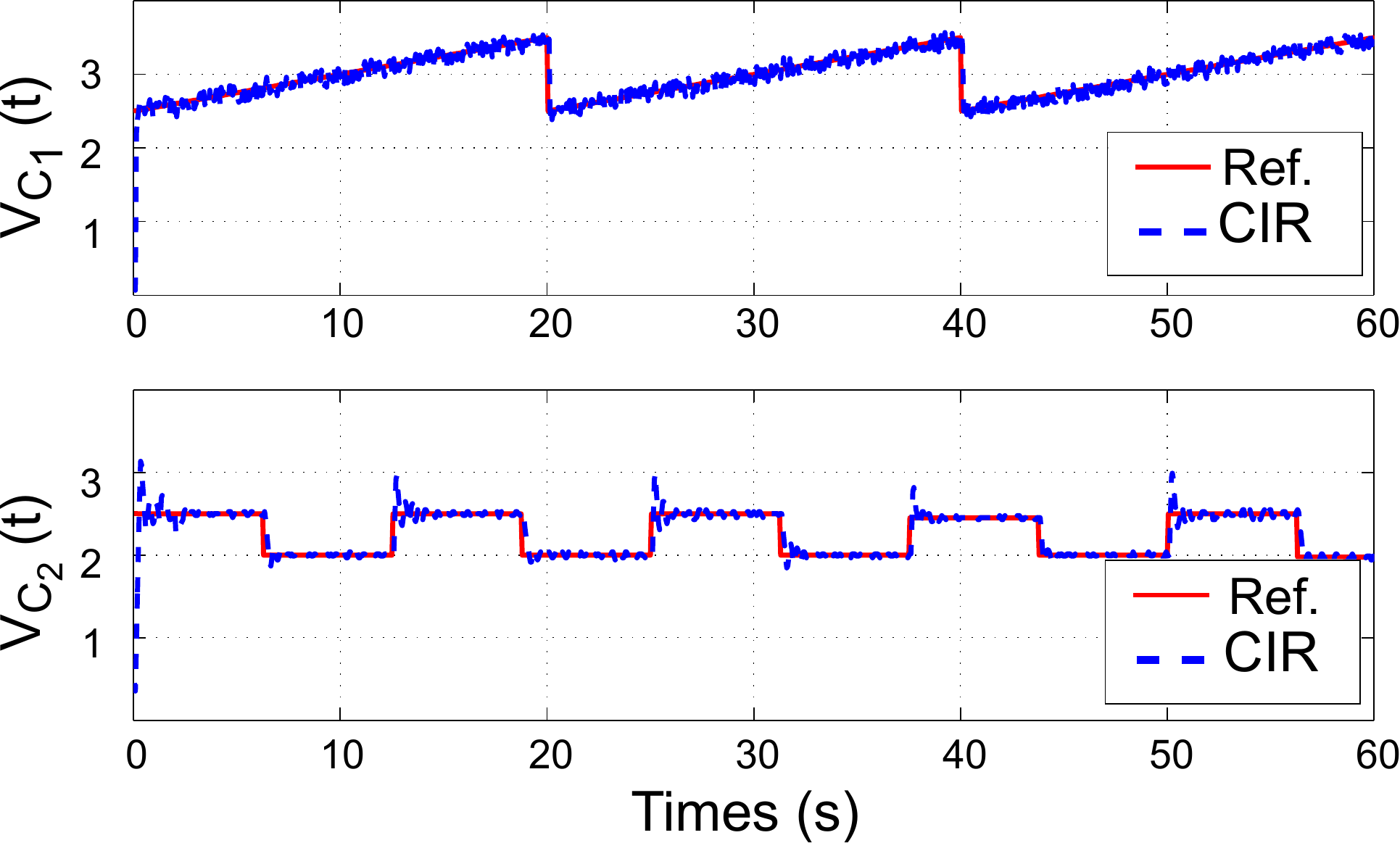}
\caption{Command following responses for the MIMO RC circuit in Example 3}
\label{fig:rc_MIMO}
\end{figure}


%
%

\begin{figure}
\centering
\includegraphics[width=0.75\linewidth]{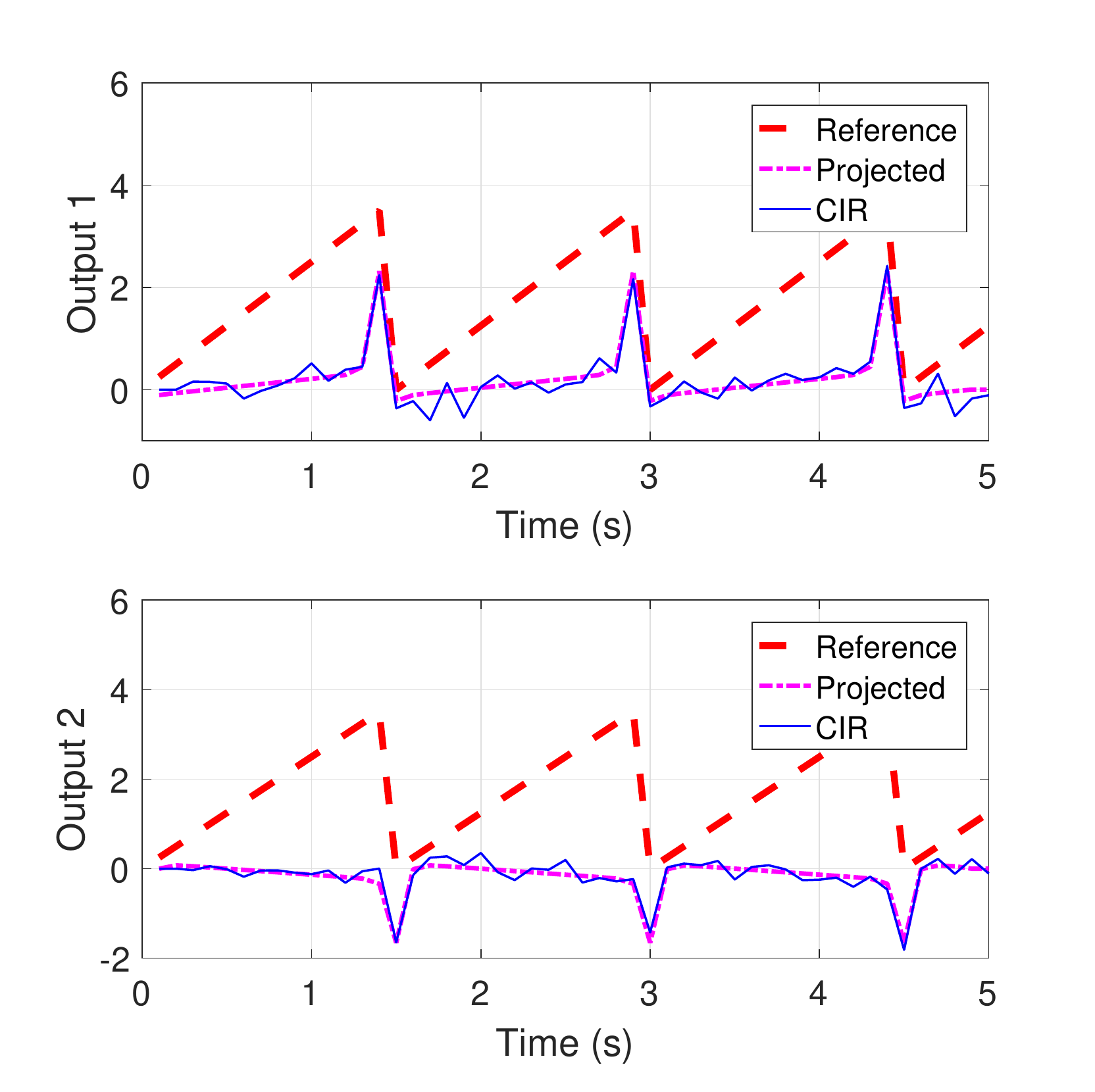}
\caption{System outputs can track the projections of the reference command on $\mathcal{R}(M_r)$}
\label{oa1}
\end{figure}

\begin{figure}
\centering
\includegraphics[width=0.75\linewidth]{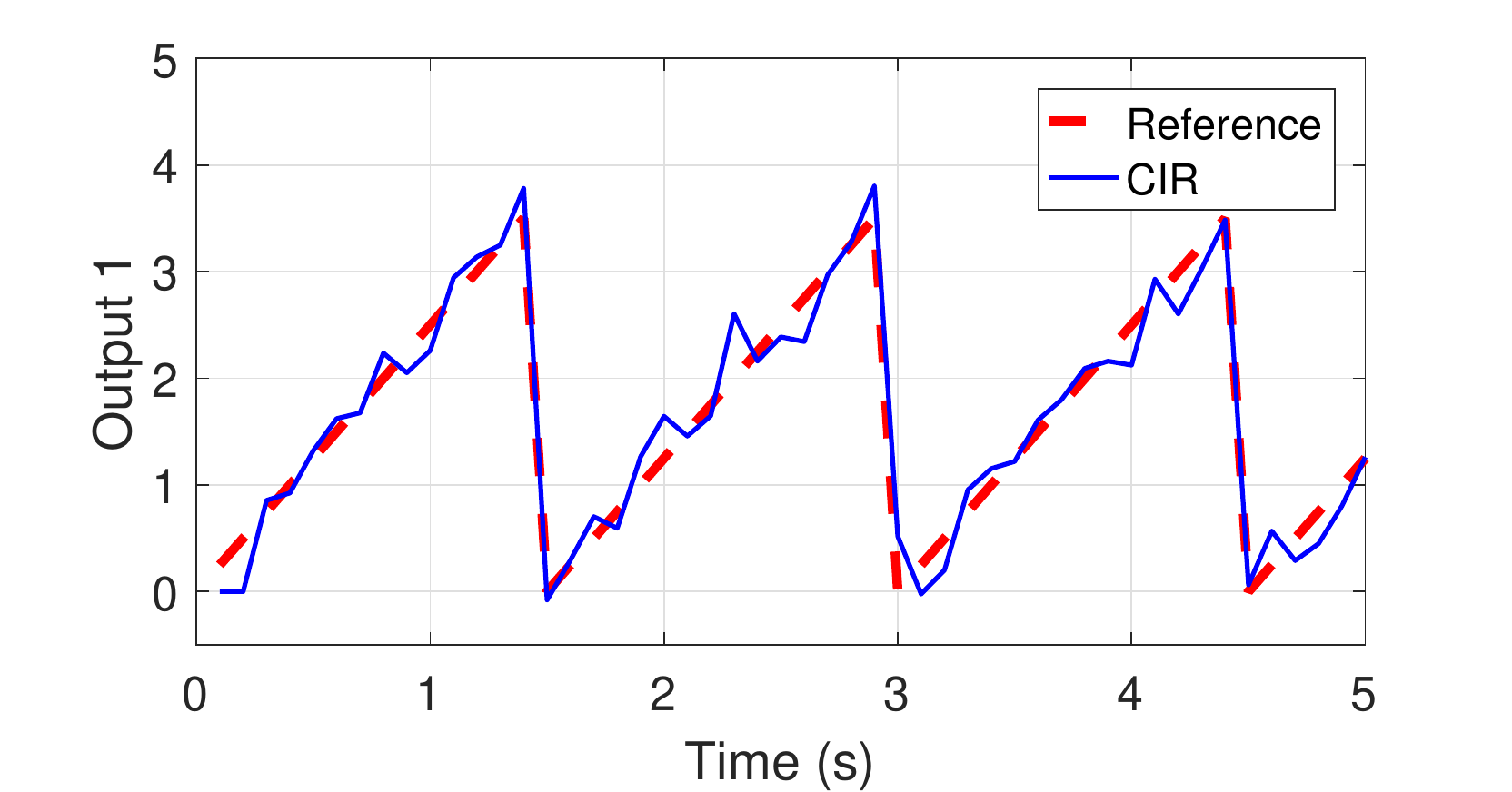}
\caption{Commands can be accurately followed for a system with $l>p$ if $p-l$ outputs are ignored keeping ${\rm rank}(CB)=l$.}
\label{oa2}
\end{figure}


%

%

\begin{thebibliography}{12}
\providecommand{\natexlab}[1]{#1}
\providecommand{\url}[1]{\texttt{#1}}
\providecommand{\href}[2]{#2}
\providecommand{\path}[1]{#1}
\providecommand{\DOIprefix}{doi:}
\providecommand{\ArXivprefix}{arXiv:}
\providecommand{\URLprefix}{URL: }
\providecommand{\Pubmedprefix}{pmid:}
\providecommand{\doi}[1]{\href{http://dx.doi.org/#1}{\path{#1}}}
\providecommand{\Pubmed}[1]{\href{pmid:#1}{\path{#1}}}
\providecommand{\BIBand}{and}
\providecommand{\bibinfo}[2]{#2}
\ifx\xfnm\undefined \def\xfnm[#1]{\unskip,\space#1}\fi
\makeatletter\def\@biblabel#1{#1.}\makeatother
\bibitem[{Xiong and Saif(2003)}]{xiong03}
\bibinfo{author}{Xiong\xfnm[ Y.]}, \bibinfo{author}{Saif\xfnm[ M.]}.
\newblock \bibinfo{title}{Unknown disturbance inputs estimation based on a
  state functional observer design}.
\newblock \emph{\bibinfo{journal}{Automatica}}
  \bibinfo{year}{2003};\bibinfo{volume}{39}:\bibinfo{pages}{1389--1398}.
\bibitem[{Hou and Patton(1998)}]{houautomatica1998}
\bibinfo{author}{Hou\xfnm[ M.]}, \bibinfo{author}{Patton\xfnm[ R.J.]}.
\newblock \bibinfo{title}{Input observability and input reconstruction}.
\newblock \emph{\bibinfo{journal}{Automatica}}
  \bibinfo{year}{1998};\bibinfo{volume}{34}(\bibinfo{number}{6}):\bibinfo{pages}{789--794}.
\bibitem[{Glover(1969)}]{glover69}
\bibinfo{author}{Glover\xfnm[ J.D.]}.
\newblock \bibinfo{title}{The linear estimation of completely unknown systems}.
\newblock \emph{\bibinfo{journal}{IEEE Trans on Automatic Contr}}
  \bibinfo{year}{1969};:\bibinfo{pages}{766--767}.
\bibitem[{Corless and Tu(1998)}]{corless98}
\bibinfo{author}{Corless\xfnm[ M.]}, \bibinfo{author}{Tu\xfnm[ J.]}.
\newblock \bibinfo{title}{State and input estimation for a class of uncertain
  systems}.
\newblock \emph{\bibinfo{journal}{Automatica}}
  \bibinfo{year}{1998};\bibinfo{volume}{34}(\bibinfo{number}{6}):\bibinfo{pages}{757--764}.
\bibitem[{Kitanidis(1987)}]{kitanidis}
\bibinfo{author}{Kitanidis\xfnm[ P.K.]}.
\newblock \bibinfo{title}{{Unbiased Minimum-variance Linear State Estimation}}.
\newblock \emph{\bibinfo{journal}{Automatica}}
  \bibinfo{year}{1987};\bibinfo{volume}{23}(\bibinfo{number}{6}):\bibinfo{pages}{775--578}.
\bibitem[{Gillijns and Moor(2005)}]{steven_kitanidis}
\bibinfo{author}{Gillijns\xfnm[ S.]}, \bibinfo{author}{Moor\xfnm[ B.D.]}.
\newblock \bibinfo{title}{Unbiased minimum-variance input and state estimation
  for linear discrete-time stochastic systems}.
\newblock \bibinfo{type}{Internal Report} \bibinfo{number}{ESAT-SISTA/TR
  05-228}; Katholieke Universiteit Leuven; \bibinfo{address}{Leuven, Belgium};
  \bibinfo{year}{2005}.
\bibitem[{Palanthandalam-Madapusi and Bernstein(2007)}]{palanthUMVACC2007}
\bibinfo{author}{Palanthandalam-Madapusi\xfnm[ H.J.]},
  \bibinfo{author}{Bernstein\xfnm[ D.S.]}.
\newblock \bibinfo{title}{Unbiased minimum-variance filtering for input
  reconstruction}.
\newblock In: \emph{\bibinfo{booktitle}{Proc. of Amer. Contr. Conf.}}
  \bibinfo{address}{New York, NY}; \bibinfo{year}{2007}:\unskip
  \bibinfo{pages}{5712 -- 5717}.
\bibitem[{Silverman(1969)}]{silverman:69}
\bibinfo{author}{Silverman\xfnm[ L.M.]}.
\newblock \bibinfo{title}{Inversion of multivariable linear systems}.
\newblock \emph{\bibinfo{journal}{Automatic Control, IEEE Transactions on}}
  \bibinfo{year}{1969};\bibinfo{volume}{14}(\bibinfo{number}{3}):\bibinfo{pages}{270--276}.
\newblock \DOIprefix\doi{10.1109/TAC.1969.1099169}.
\bibitem[{D'Amato(2013)}]{damato2013}
\bibinfo{author}{D'Amato\xfnm[ A.M.]}.
\newblock \bibinfo{title}{Minimum-norm input reconstruction for
  nonminimum-phase systems}.
\newblock In: \emph{\bibinfo{booktitle}{Proc. of Amer. Contr. Conf.}}
  \bibinfo{address}{Washington, DC}; \bibinfo{year}{2013}:\unskip
  \bibinfo{pages}{3111 -- 3116}.
\bibitem[{Sain and Massey(1969)}]{sain69}
\bibinfo{author}{Sain\xfnm[ M.K.]}, \bibinfo{author}{Massey\xfnm[ J.L.]}.
\newblock \bibinfo{title}{Invertibility of linear time-invariant dynamical
  systems}.
\newblock \emph{\bibinfo{journal}{IEEE Trans on Automatic Contr}}
  \bibinfo{year}{1969};\bibinfo{volume}{AC-14}(\bibinfo{number}{2}):\bibinfo{pages}{141--149}.
\bibitem[{Marro(1992)}]{marroGA}
\bibinfo{author}{Marro\xfnm[ G.]}.
\newblock \bibinfo{title}{Controlled and Conditioned Invariants in Linear
  System Theory}.
\newblock \bibinfo{publisher}{Prentice Hall}; \bibinfo{year}{1992}.
\bibitem[{Kadam and Palanthandalam-Madapusi(2017)}]{kadam2017revisiting}
\bibinfo{author}{Kadam\xfnm[ S.D.]},
  \bibinfo{author}{Palanthandalam-Madapusi\xfnm[ H.J.]}.
\newblock \bibinfo{title}{Revisiting trackability of linear time-invariant
  systems}.
\newblock In: \emph{\bibinfo{booktitle}{American Control Conference, 2017}}.
  \bibinfo{organization}{IEEE}; \bibinfo{year}{2017}:\unskip
  \bibinfo{pages}{(submitted)}.

\end{thebibliography}
\end{document}